\documentclass[11pt,letterpaper,USenglish]
{article}
\usepackage{amsthm,color,latexsym,graphicx}
\usepackage[hyphens]{url}
\usepackage[margin=1in]{geometry}
\usepackage[font=small,labelfont=bf]{caption}
\usepackage[labelformat=simple]{subcaption}

\DeclareSymbolFont{matha}{U}{matha}{m}{n}

\DeclareMathSymbol{\Lt}{3}{matha}{"CE}
\DeclareMathSymbol{\Gt}{3}{matha}{"CF}

\usepackage{float}

\usepackage{indentfirst,csquotes}

\usepackage{paralist,etoolbox}
\usepackage{
		graphicx,		
		amsmath,		
		amssymb,		
		amsthm,			
		tikz,			
		multicol,		
		color,			
		thmtools,		
		etoolbox,		
		xspace,			
        hyperref,		
        xcolor,         
        datetime,       
        csquotes,        
        svg,
        placeins
        }
        
\usepackage{hyperref}
\hypersetup{breaklinks,
bookmarksnumbered,bookmarksopen,bookmarksopenlevel=2}
{\makeatletter \hypersetup{pdftitle={\@title}}}

\newtheorem{theorem}{Theorem}[]
\newtheorem*{theorem-no-number}{Theorem}

\newtheorem{lemma}[theorem]{Lemma}
\newtheorem{proposition}[theorem]{Proposition}
\newtheorem{corollary}[theorem]{Corollary}

\usepackage{float}

\usepackage{indentfirst,csquotes}

\usepackage{paralist,etoolbox}

\usepackage{comment}
\usepackage{doi}

\usepackage{afterpage}
\usepackage{placeins}
\usepackage{amsfonts}
\usepackage{graphicx}
\usepackage{hyperref}
\usepackage{color}

\DeclareFontFamily{U}{matha}{\hyphenchar\font45}
\DeclareFontShape{U}{matha}{m}{n}{
      <5> <6> <7> <8> <9> <10> gen * matha
      <10.95> matha10 <12> <14.4> <17.28> <20.74> <24.88> matha12
      }{}
\DeclareSymbolFont{matha}{U}{matha}{m}{n}

\DeclareMathSymbol{\Lt}{3}{matha}{"CE}
\DeclareMathSymbol{\Gt}{3}{matha}{"CF}

\usepackage{subcaption}

\newcommand{\gpos}{$G_{pos}\textrm{(POS CNF)}$}

\begin{document}
\title{Quoridor is PSPACE-Complete}
%
%

\author{%
  Marius Drop
\and
  Benjamin G. Rin%
    \thanks{Corresponding author. Utrecht University, Utrecht, The Netherlands.
      \protect\url{b.g.rin@uu.nl}}
\and
  Finn van der Velde
    }
    
\date{March 15, 2026}


%
\maketitle              

\begin{abstract}
Quoridor is an award-winning abstract strategy game designed by Mirko Marchesi and published in 1997. Similar games include Maze Attack, Blockade (also known as \textit{Cul-de-sac}), and Pinko Pallino.   
In line with chess, checkers, Go, and other classic combinatorial games, Quoridor is a turn-based, deterministic, perfect-information game  
played on a  
square grid.  
We show that  
it is PSPACE-complete to determine whether a given player has a winning strategy in a given Quoridor position on a board with size~$n\times n$. We prove this  
by reduction from \gpos, a  
Boolean formula game  
originally
defined in 1978 by T.~Schaefer.


\end{abstract}


\section{Introduction}

Quoridor is a strategy game designed in 1997 by Mirko Marchesi. It is closely related to predecessors Blockade (also known as \textit{Cul-de-sac}), designed in 1975 by Philip Slater, and Pinko Pallino, designed in 1995 also by Mirko Marchesi. A more recent game, Maze Attack, has similar rules. The focus of the present paper is on Quoridor. Differences among these games and their effects on this paper's results are discussed in Section~\ref{sec:disc}.

In Quoridor, players each have a \emph{pawn} that is able to move up to one square per turn. They also have \textit{walls} that can be placed between spaces to restrict the pawns' movement. The goal is to be the first player whose pawn reaches the opposite side of the board. In this respect, the game resembles some older games such as Hex, which has been known since~\cite{reisch1981hex} to be PSPACE-complete. The aim of the present work is to determine the computational complexity of Quoridor. Our main finding is that the problem of determining whether a given player has a winning strategy in a given position in (an appropriately generalized version of) Quoridor is likewise PSPACE-complete. This resolves a problem that has been open for nearly two decades~\cite{Demaine-Hearn2009-paper}.
We achieve this by reduction from T.~Schaefer's formula game~\gpos. Similar types of results, using a wide variety of proof strategies, have been found 
for
games such as chess~\cite{FraenkelLichtenstein, storer}, checkers~\cite{RobsonCheckers},  
Go~\cite{gopspacehard, RobsonGO},   
Amazons~\cite{Buro-Amazons, 
amazonsKonaneCrosspurposesPSPACEcomplete}, Othello/Reversi~\cite{iwata1994othello}, Phutball~\cite{Phutball}, Havannah~\cite{BonnetJS14}, Arimaa~\cite{Rin-Schipper24}, Gomoku~\cite{Gobang}, and many others.

We begin with a description of the game rules (Section~\ref{sec:rules}) and an overview of the results (Section~\ref{sec:results}). We then lay out a few preliminaries (Section~\ref{sec:prelim}), followed by the reduction proof itself (Section~\ref{sec:proof}). We conclude with a short discussion (Section~\ref{sec:disc}).

\subsection{Game rules} \label{sec:rules}
On a $9 \times 9$ square grid, opposite-facing players White and Black each begin with a pawn  
in the center of their back rank. Each player also begins with ten \emph{wall} pieces in their supply. In a given turn, a player either moves their pawn or places a wall piece on the board. Pawns can move one square cardinally, or jump over an opposing pawn if cardinally adjacent (see Fig.~\ref{jumpingstraight}). Additional rules about jumping are discussed below. A player wins the game when their pawn is the first to reach the opposite edge of the grid.

Wall pieces, which have length~2 and width~0, are placed on the borders between squares. A pawn cannot pass through a wall.
Once placed, a wall can never be removed or relocated. Walls cannot be placed to overlap, nor can they intersect to form a plus.  Importantly, it is illegal to place a wall such that 
it leaves 
a player
with no possible path to victory. Thus, the game is not about permanently trapping one's opponent, but slowing them down long enough to win. 

When jumping over an adjacent opposing pawn with an available space behind it, 
the jump occurs as in Fig.~\ref{jumpingstraight}.
If the space behind it is instead blocked by a wall or the edge of the board, 
the jumping pawn must land to the side of the opponent’s pawn (see Fig.~\ref{jumpingtwice}). Another wall placed to the side of the opponent's pawn may block this  
(Fig.~\ref{jumpingonce}).

\begin{figure}[tbp!]
  \centering
  \begin{minipage}[b]{0.3\textwidth}
    \includegraphics[width=\textwidth]{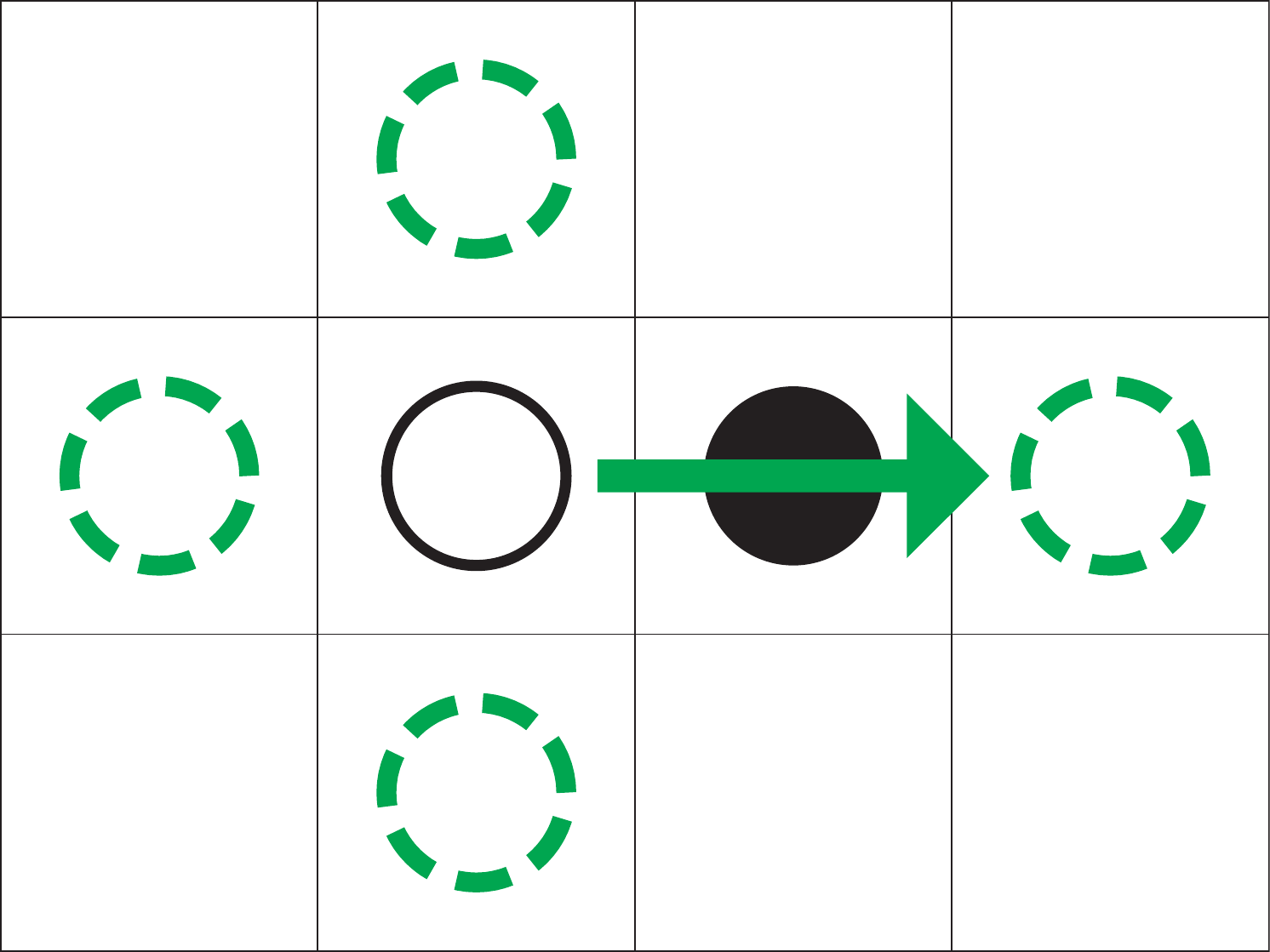}
    \caption{Jumping to the open square behind an opponent's pawn.}
    \label{jumpingstraight}
  \end{minipage}
  \hfill
  \begin{minipage}[b]{0.3\textwidth}
    \includegraphics[width=\textwidth]{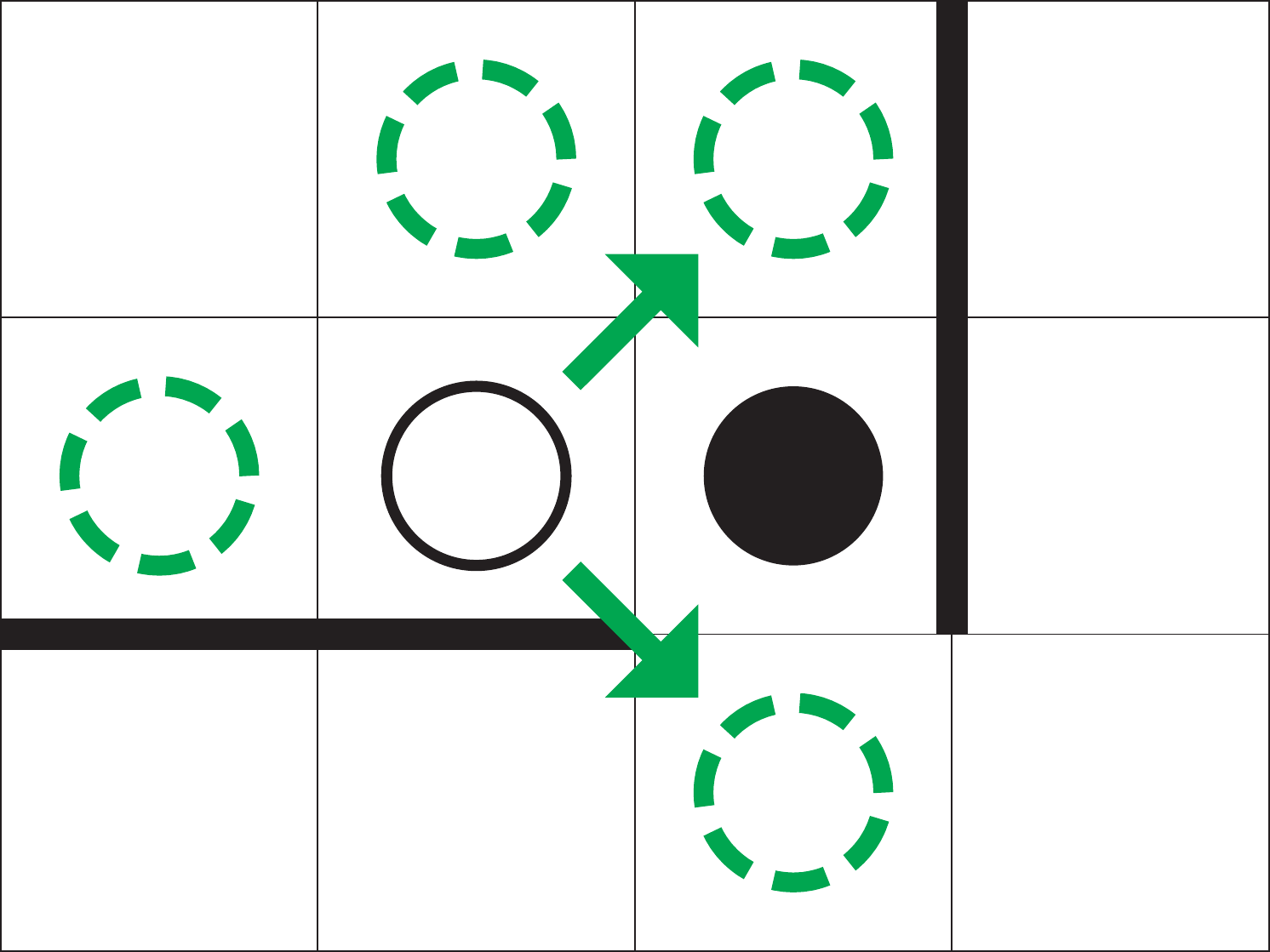}
    \caption{Jumping with a wall behind an opponent's pawn.\\ \ }
    \label{jumpingtwice}
  \end{minipage}
  \hfill
  \begin{minipage}[b]{0.3\textwidth}
    \includegraphics[width=\textwidth]{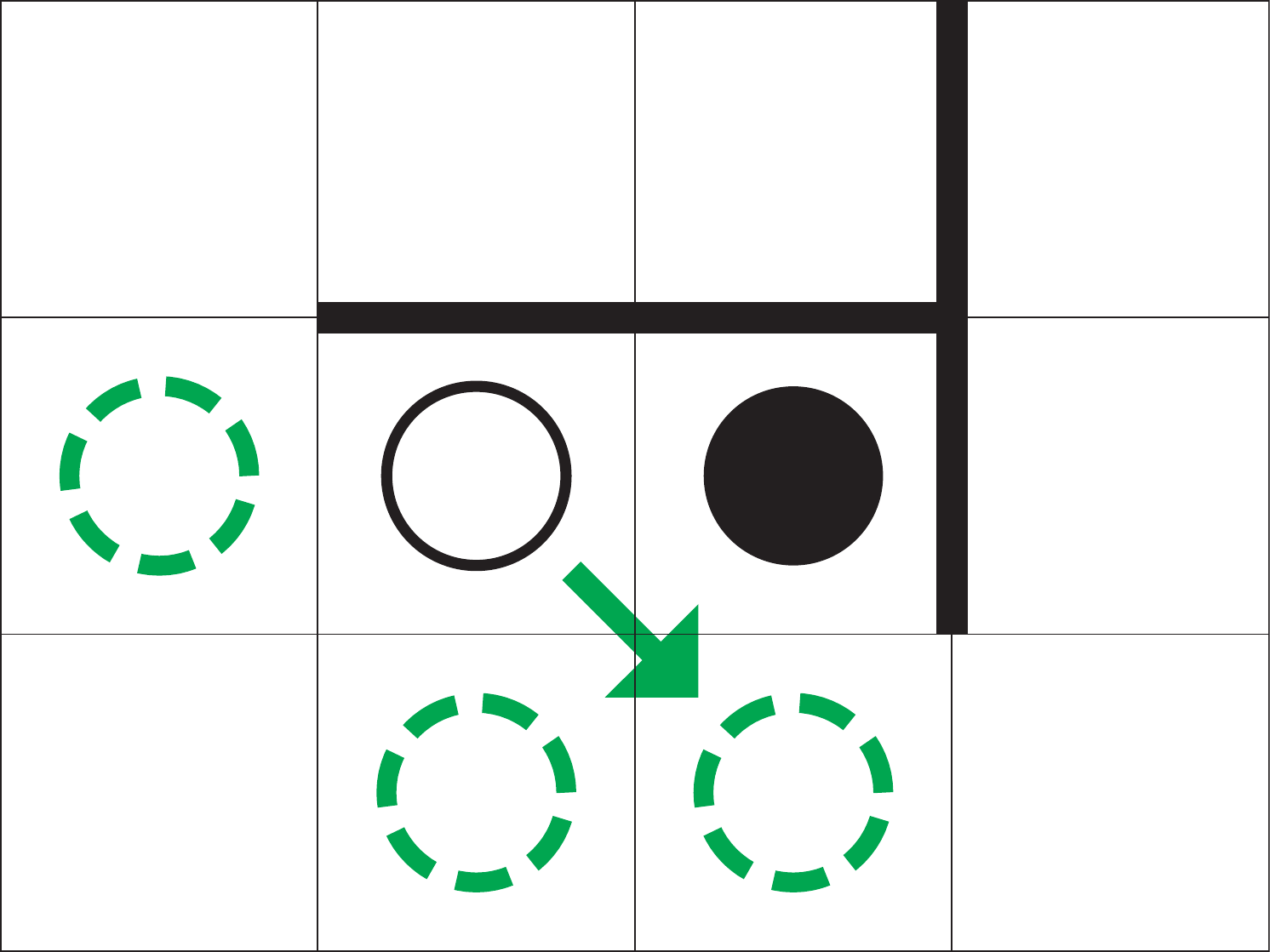}
    \caption{Jumping with walls behind and next to an opponent's pawn.}
    \label{jumpingonce}
  \end{minipage}
\end{figure}

Note that Quoridor has no explicit rule limiting the length of games in the face of endlessly repeated positions, such as the three-fold repetition rule in chess.

As a strategic consequence of the rule requiring that each player always have a path to victory, there are times when it is rational for a player to place a wall blocking one of their \textit{own}
paths to victory. This follows from the fact that a player with multiple victory paths who carelessly moves their pawn along one of them may have an opponent block it off just before the pawn reaches the end, forcing the player to waste time retracing their steps and choose another path. 
By blocking all but one of one's own victory paths, a player can ensure
that 
their final path is 
legally unblockable before committing to a given direction. This scenario arises for Black in the proof of our main theorem. Indeed, his\footnote{For ease of exposition, we have elected to use
`he' for Black and `she' for White throughout this paper.}  strategy hinges on it.

\subsection{Results}\label{sec:results}

This paper considers a generalization of Quoridor to an $n \times n$ grid, with a reasonable number  
of wall pieces per player  
(say, scaling   
linearly with the 
grid size). For example, we may  
stipulate that they each have $\frac{10n^2}{81}$ wall pieces,  
resulting in ten wall pieces per player for every 81 squares, as in the original game.
Players still have one pawn each. The \textsc{Quoridor} decision problem 
is then defined as
the problem of
determining whether White 
has a  
winning strategy in an arbitrary  
position of this generalized version of the game.\footnote{One may also wish to consider the complexity of the game as played from its starting position. This question is not addressed in the present paper.} Our main result is  
the following.

\begin{theorem-no-number}
The \textsc{Quoridor} decision problem is PSPACE-complete.\footnote{The proof presented in this paper originates from work done for the bachelor thesis of the alphabetically third listed author, written under supervision of the second.~\cite{FinnThesis}}
\end{theorem-no-number}

Showing
\textsc{Quoridor} $\in$ PSPACE  is straightforward (see Theorem~\ref{thm:inPSPACE}).
To show hardness, we give a reduction from \gpos, a PSPACE-complete Boolean formula game introduced in 1978 by Schaefer:
 \begin{quote}
Input is a positive CNF formula (that is, a CNF formula~[denoted $A$] in which $\lnot$ does not occur). A move consists of choosing some variable of~$A$ which has not yet been chosen. The game ends after all variables of~$A$ have been chosen. Player~I wins if and only if~$A$ is true when all variables chosen by~I are set to~1 and all variables chosen by~II are set to~0. (In other words, [player]~I wins if and only if [player I] succeeds in playing some variable in each conjunct.) For example, on input $x \land (y \lor z) \land (y \lor w)$ player~II can win. \cite{Schaefer} \end{quote}
The problem of \gpos\ has seen attention in the study of Maker-Breaker games, some recent findings on which can be found in, e.g.,~\cite{4-Uniform,Koepke2025,6-Uniform}.
In proofs of PSPACE-hardness, \gpos\ and related problems have found some moderate use for various reductions (cf.~\cite{Tumbleweed,DotsAndBoxes, miserepartizanarckayles,Duchene-et-al-2025,Fenner-et-al-2015, amazonsKonaneCrosspurposesPSPACEcomplete}. What distinguishes \gpos\ from problems such as TQBF is that players do not choose which truth-value to assign a variable on their turns, but they do choose which variable to select. 

Rather than having the game end immediately when all variables are 
chosen, the present paper adopts an equivalent formulation of \gpos\ wherein Player~II chooses any one clause, then Player~I chooses a literal within that clause, after which Player~I wins if and only if the chosen literal is {\tt true}. 

While Quoridor can normally be played with up to four players, our reduction requires only two. 
If desired, the proofs of our results 
can be easily adapted to accommodate three or four players (see Sec.~\ref{sec:disc})

\begin{theorem}\label{thm:inPSPACE}
    {\sc Quoridor} $\in$ PSPACE.
\end{theorem}

\begin{proof}
  Observe that there are~$n^2 \times n^2 = n^4$ unique pairs of White and Black pawn locations, putting a polynomial bound on the number of consecutive pawn moves possible between wall-placing turns without repeating a position. The total number of wall pieces available is bounded above by~$n^2$, so without repeating positions, a game's length is bounded by~$n^6$.

     An algorithm can construct a game tree that prunes any branch with a White move that repeats an earlier position, since such a move can never be the best White move unless all options fail to force a win (in which case, pruning this branch will anyway not lead to a wrong answer, as the algorithm will correctly answer ``no'' after evaluating the other branches). Because every wall-placing turn guarantees that future positions are distinct from positions before the wall was placed, no more than~$n^4$ positions need to be stored in memory to check whether a potential White move creates a repetition. So, since the height of the pruned decision tree is polynomial, the decision problem is in PSPACE (see, e.g.,~\cite[Lem. 2.2]{Schaefer}).
\end{proof}

\begin{theorem}\label{thm:PSPACE-hard}
    {\sc Quoridor} is PSPACE-hard.
\end{theorem}

Our proof will give a reduction that maps any instance of \gpos\ to a position in $(n\times n)$ Quoridor 
preserving
the existence or lack of a winning strategy for Player~I (represented by White). Note, however, that negative answers to the {\sc Quoridor} decision problem include both positions in which Black has a winning strategy and positions in which neither player has a winning strategy (where the best option for both sides is to repeat moves). We therefore have elected to add an extra component to our reduction---the \textit{black victory corridors}---to remove the latter possibility, at the cost of mildly increasing the reduction's complexity. This enables us to establish the following (slightly stronger) result.

\begin{corollary}\label{cor:main}
    {\sc Quoridor} is PSPACE-complete even when restricted to positions in which 
    either White or Black has a forced winning strategy.
\end{corollary}

\section{Preliminary constructions and definitions}\label{sec:prelim}

\subsection{Corridors}

A \textbf{\textit{corridor}} consists of two parallel chains of walls separated by distance~1. Additionally, walls perpendicular to each of these walls are placed outside the corridor to ensure that no further walls can legally be added into the corridor to block it (recall that walls cannot overlap). Some examples of such perpendicular walls include the two uppermost vertical walls in Figs.~\ref{fig:railroading-T1}-\ref{fig:railroading-T2} and the two rightmost horizontal walls in Figs.~\ref{truth-value-assignment1}-\ref{truth-value-assignment3}. 

\subsection{Railroading}
While not typical in real $9\times 9$ Quoridor games, a situation frequently arises in our proof wherein Black's pawn stands in a junction of two corridors, either a T-junction or cross (4-way) junction, legally preventing White from making a $90^\circ$ turn through it for as long as Black's pawn stays there. We refer to this behavior as \textbf{\textit{railroading}}. See Figs.~\ref{fig:railroading-T1}-\ref{fig:railroading-cross}. In the case of a T-junction, railroading prevents 
making $90^\circ$~turns into the stem of the T,  
but not
turns originating from the stem. At a cross junction, railroading prevents all $90^\circ$ turns.

\begin{figure}[h!]
  \centering
  \begin{minipage}[b]{0.3\textwidth}
  \begin{center}
    \includegraphics[width=\textwidth]{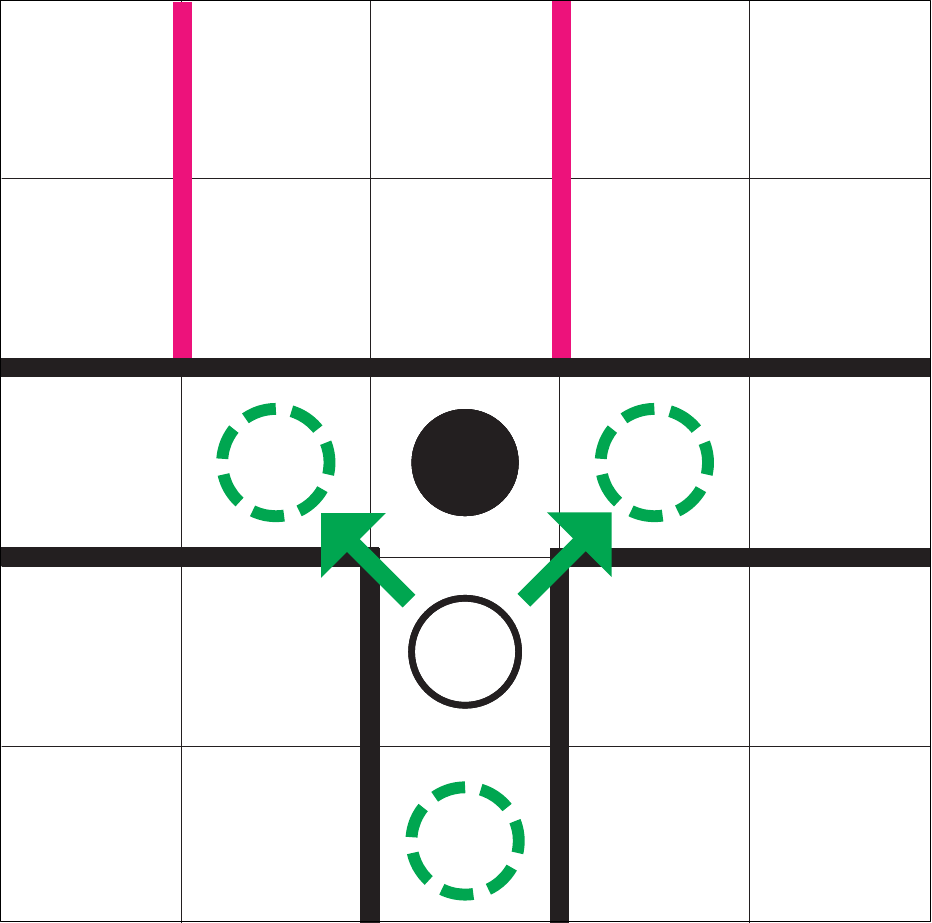}
    \caption{White's options when moving at a railroaded T-junction from the stem. Turning $90^\circ$ is possible.}
    \label{fig:railroading-T1}
   \end{center}
  \end{minipage}
  \hfill
  \begin{minipage}[b]{0.3\textwidth}
    \begin{center}
    \includegraphics[width=\textwidth]{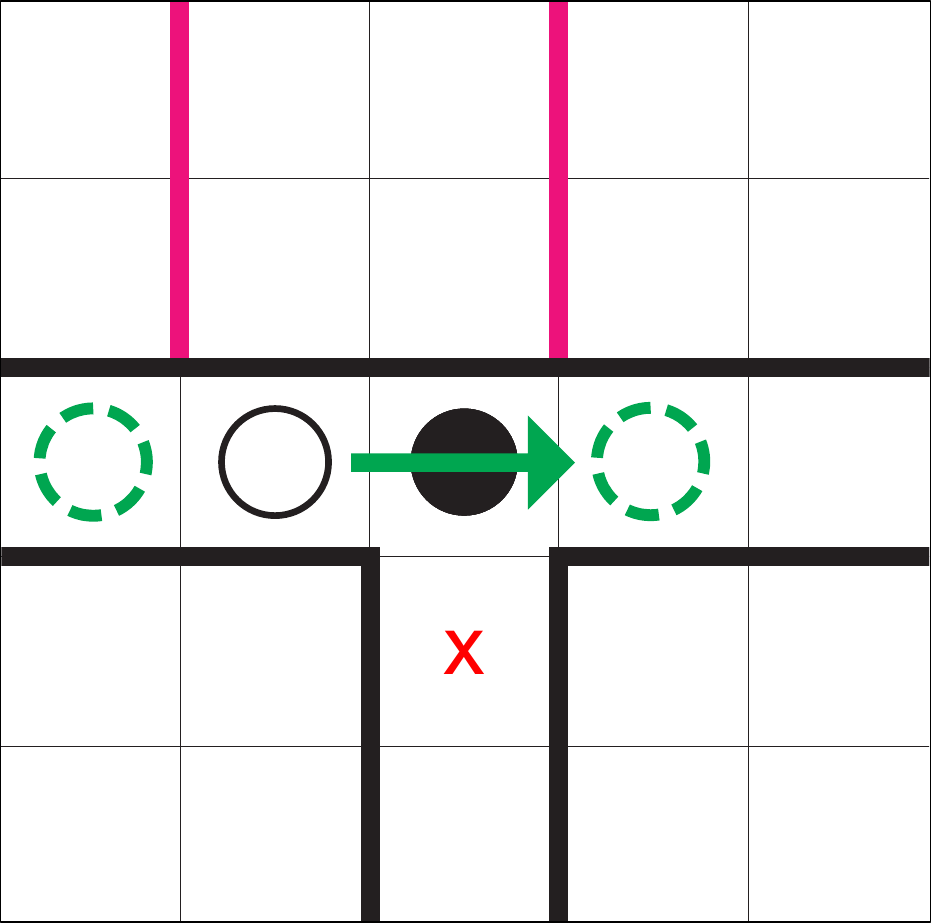}
    \caption{White's options at a railroaded T-junction not starting from the stem. Turning $90^\circ$ is impossible.}
    \label{fig:railroading-T2}
    \end{center}
  \end{minipage}
  \hfill
  \begin{minipage}[b]{0.303\textwidth}
  \begin{center}
    \includegraphics[width=\textwidth]{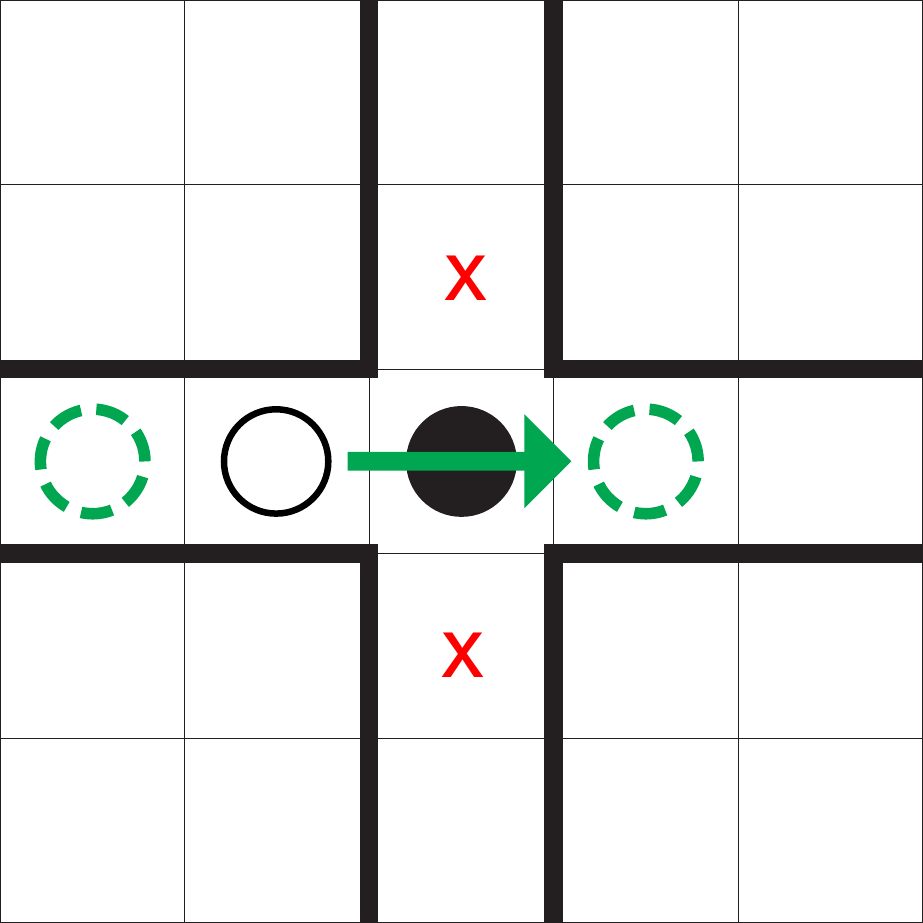}

    
    \caption{White's options moving at a railroaded cross junction from any direction. Turning $90^\circ$ is
    impossible.}
    \label{fig:railroading-cross}
    \end{center}
  \end{minipage}
  
\end{figure}

\subsection{Unblockable T-junctions}
It will be useful in our proof to have some T-junctions where railroading is impossible. Such an \textbf{\textit{unblockable T-junction}}, depicted in Fig.~\ref{fig:unblockable-T-junction}, is functionally a T-junction through which White can move freely regardless of where Black stands. This construction can be rotated to any orientation. All T-junctions in the proof will be unblockable by default unless otherwise indicated.  Unblockable cross junctions are also possible to construct, but not needed or desired in the proof, as we intend railroading to be possible there. 

\begin{figure}[ht]
\centering
    \includegraphics[width=.66\textwidth] 
    {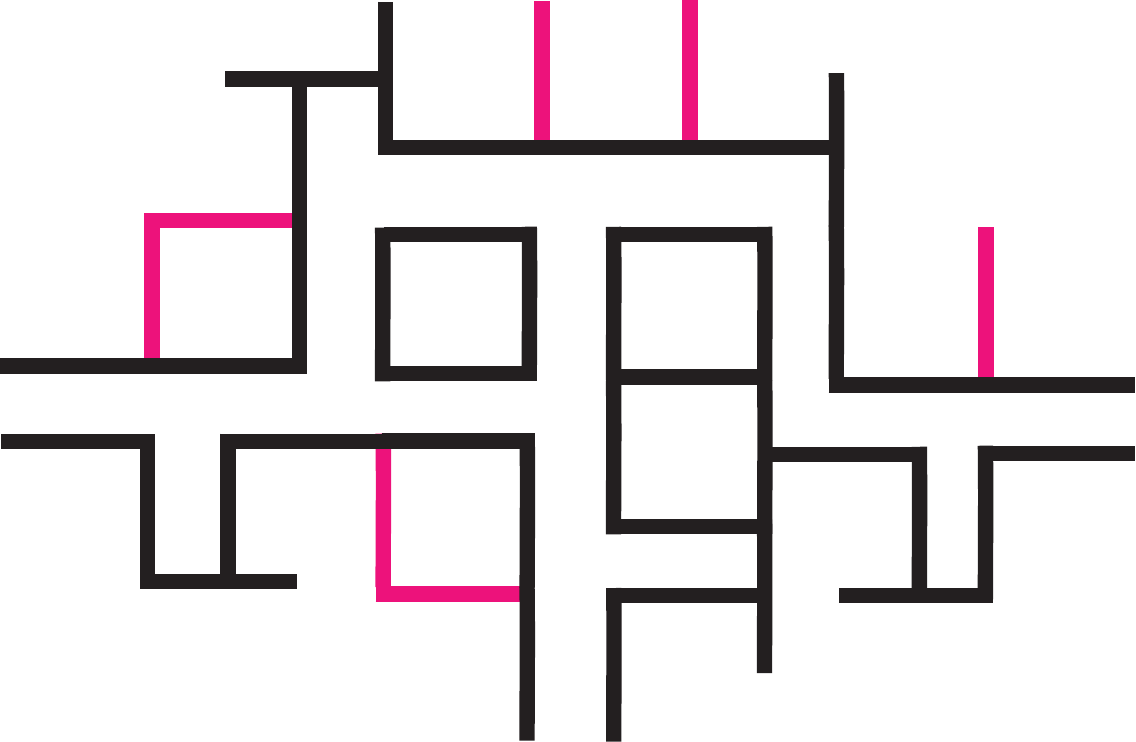}
    \caption{An unblockable T-junction. 
    White can travel freely from any entrance to any other, regardless of the Black pawn's 
       location. 
    }
    \label{fig:unblockable-T-junction}
\end{figure}

\section{Proof}\label{sec:proof}

We first sketch the proof and then provide the details.

\begin{figure}[h]
    \includegraphics[width=\textwidth] 
    {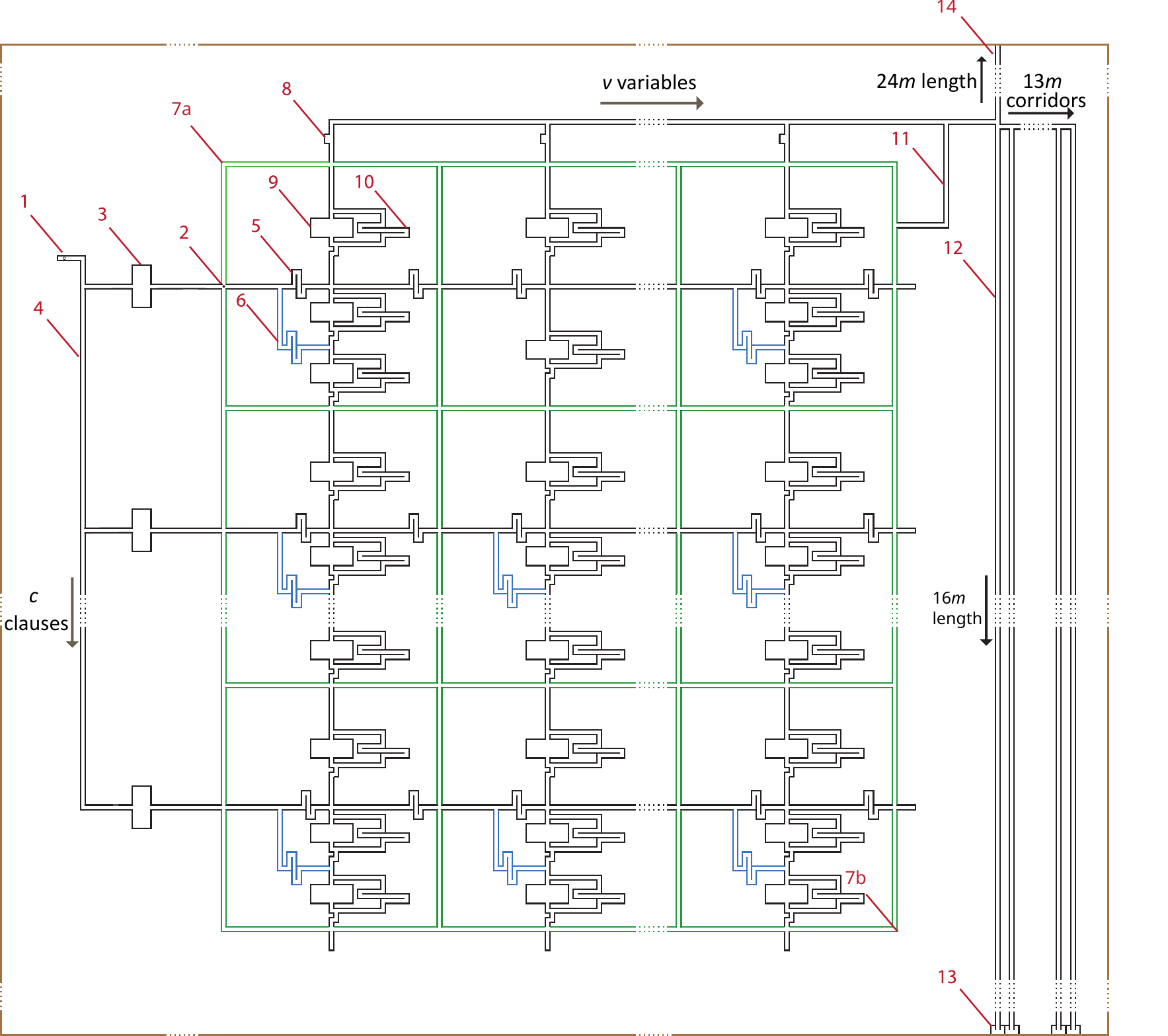}
    \caption{The reduction position at a high level (some details omitted). Image is not drawn to scale. The outer edge represents 
    the edge of the Quoridor board. Labeled items are as follows.\\
       {\color{red}1.}~White starting location. {\color{red}2.}~Black starting location. {\color{red}3.}~Clause entrance chamber. {\color{red}4.}~Chamber access corridor. {\color{red}5.}~Winding road to increase corridor to a specified length~$L$. {\color{red}6.}~Side path for White to force entry into variable gadget, placed by the intersection of every clause and variable gadget for which  
        the corresponding 
        clause in the \gpos\ formula contains 
        the variable. The winding road within increases the length to~$L$. {\color{red}7a.}~The maze (top-left corner). {\color{red}7b.}~The maze (bottom-right corner). {\color{red}8.}~Truth-value assigning box. {\color{red}9.}~Chamber inside variable gadget. {\color{red}10.}~Long and winding road with length $3L$. {\color{red}11.}~Emergency exit accessible through the maze. {\color{red}12.}~Black’s corridors to White’s first rank. {\color{red}13.}~Locations where Black’s corridors to White’s first rank can be walled off. 
        {\color{red}14.}~White's corridor to~Black's~first~rank.}
    \label{fig:labeledreduction}
\end{figure}

\afterpage{\FloatBarrier}

\subsection{Proof outline}
We will reduce \gpos\ to {\sc Quoridor} as follows. Given positive CNF formula~$A$ with~$c \in \mathbb N$ clauses 
and~$v\in \mathbb N$ variables
, we construct a game position as depicted in Fig.~\ref{fig:labeledreduction},  whose details are explained in the proof. The construction consists primarily of a \textit{\textbf{maze}} from which the white pawn must escape in order to reach
a path to 
the top rank and win. The number~$m$ of squares in the maze is determined by a polynomial function of~$v$ and~$c$ discussed later.
The maze is constructed such that White is able to force an escape and win the game if only if Player~I has a winning \gpos\ strategy for $A$.

To each clause (variable) in $A$ there corresponds a \textit{\textbf{row (column)}} in the maze. Note that we have used the term \textit{rank (file)} to denote a straight series of horizontally (vertically) connected squares on the board, whereas the term \textit{\textbf{row (column)}} used here denotes a horizontal (vertical) strip of the maze occupying many ranks (files). We call the set of squares comprising the intersection of one row and column a \textit{cell.}

One exit out of the maze stands at the top of each column.\footnote{The maze also has a single additional exit---the \textit{emergency} exit---described below.}
Each variable's truth assignment is 
represented by a $2 \times 2$ \textit{\textbf{box}} at the top for assigning its 
value (labeled~8). White (Black) simulates Player~I~(II) choosing a variable in $A$ by 
placing 
a vertical (horizontal) wall in the variable's box, representing a {\tt true} {(\tt false)} assignment (see Figs.~\ref{truth-value-assignment1}-\ref{truth-value-assignment3}). 
Note that since walls cannot form a plus, only one wall can fit in each box. 
Thus a vertical (horizontal) wall permanently opens (closes) the exit.
We will see that if Black can assign all variables in one clause {\tt false}, the construction of the maze allows him to delay White's exit indefinitely, for as long as Black himself remains in the maze. This is accomplished by continually railroading White so she cannot turn north, no matter which junction she visits, by always arriving at any junction faster than White because of the many zig-zagging path in the maze that slow White down. We will see that this strategy buys Black sufficient time to wall off~$b-1$ of his own~$b$-many corridors to victory (see Fig.~\ref{fig:labeledreduction}), where~$b:=13m$ is a 
number large relative to the size~$m$ of the maze. 
Doing so leaves Black with
a unique (therefore, unblockable) winning path~$16m$ steps long to the board's bottom rank. 
In this case, Black wins faster than White, whose only path to victory is~$24m$ steps. However, if Black cannot do this---i.e., every clause is {\tt true}---then White has a way to escape the maze and use 
her~$24m$-length corridor while Black still has multiple open corridors of length~$16m$. Black then loses, as White can legally wall off whichever corridor Black enters, one turn before Black reaches the bottom, forcing him to backtrack and walk 
at least 
$16m+16m=32m$ more steps to try futilely 
to win the race via another corridor.

\begin{figure}[h!]
  \centering
  \begin{minipage}[b]{0.3\textwidth}
  \begin{center}
    \includegraphics[width=\textwidth]{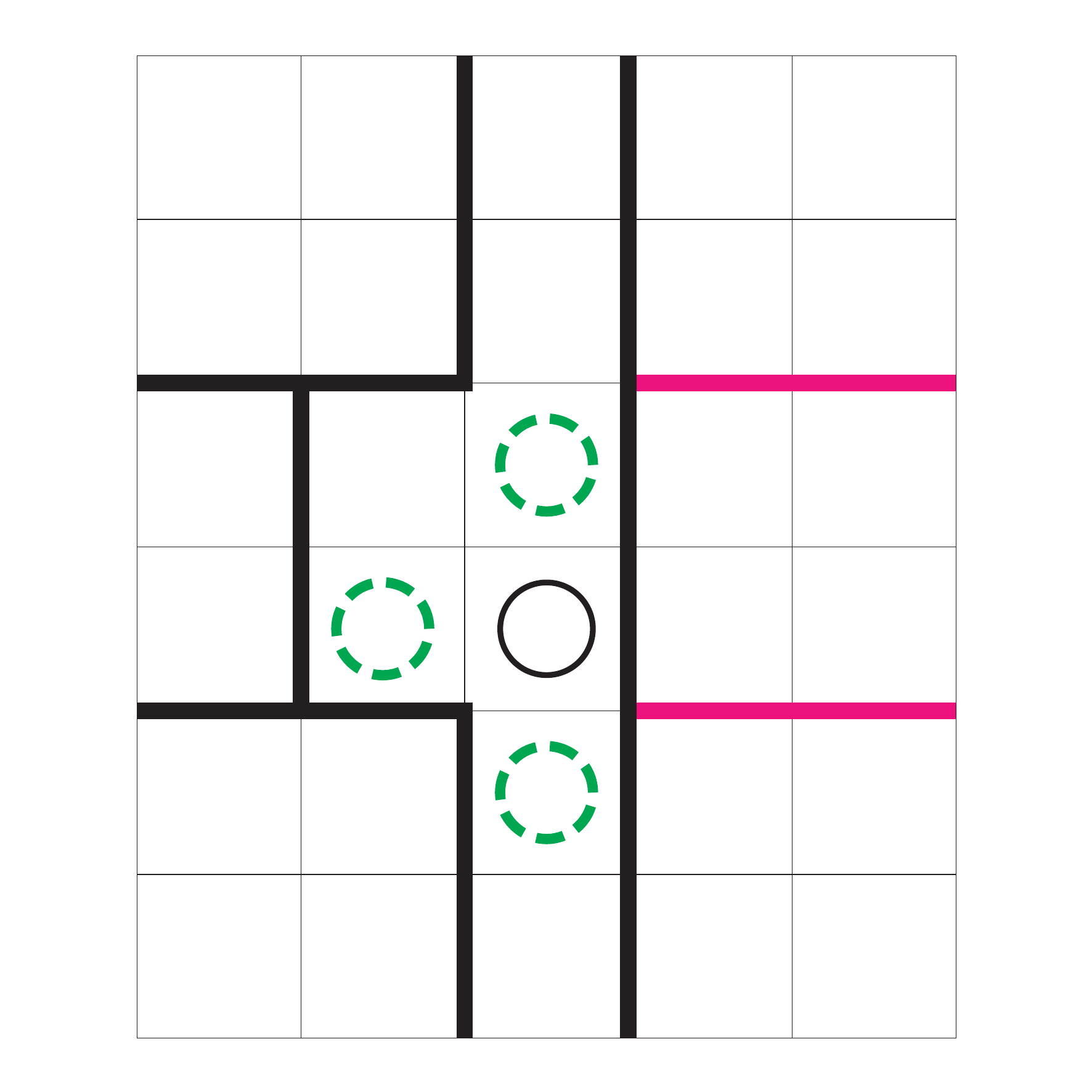}
    \caption{No new wall in the $2\times 2$ area with the pawn. The pawn has three squares to move to.}
    \label{truth-value-assignment1}
   \end{center}
  \end{minipage}
  \hfill
  \begin{minipage}[b]{0.3\textwidth}
    \begin{center}
    \includegraphics[width=\textwidth]{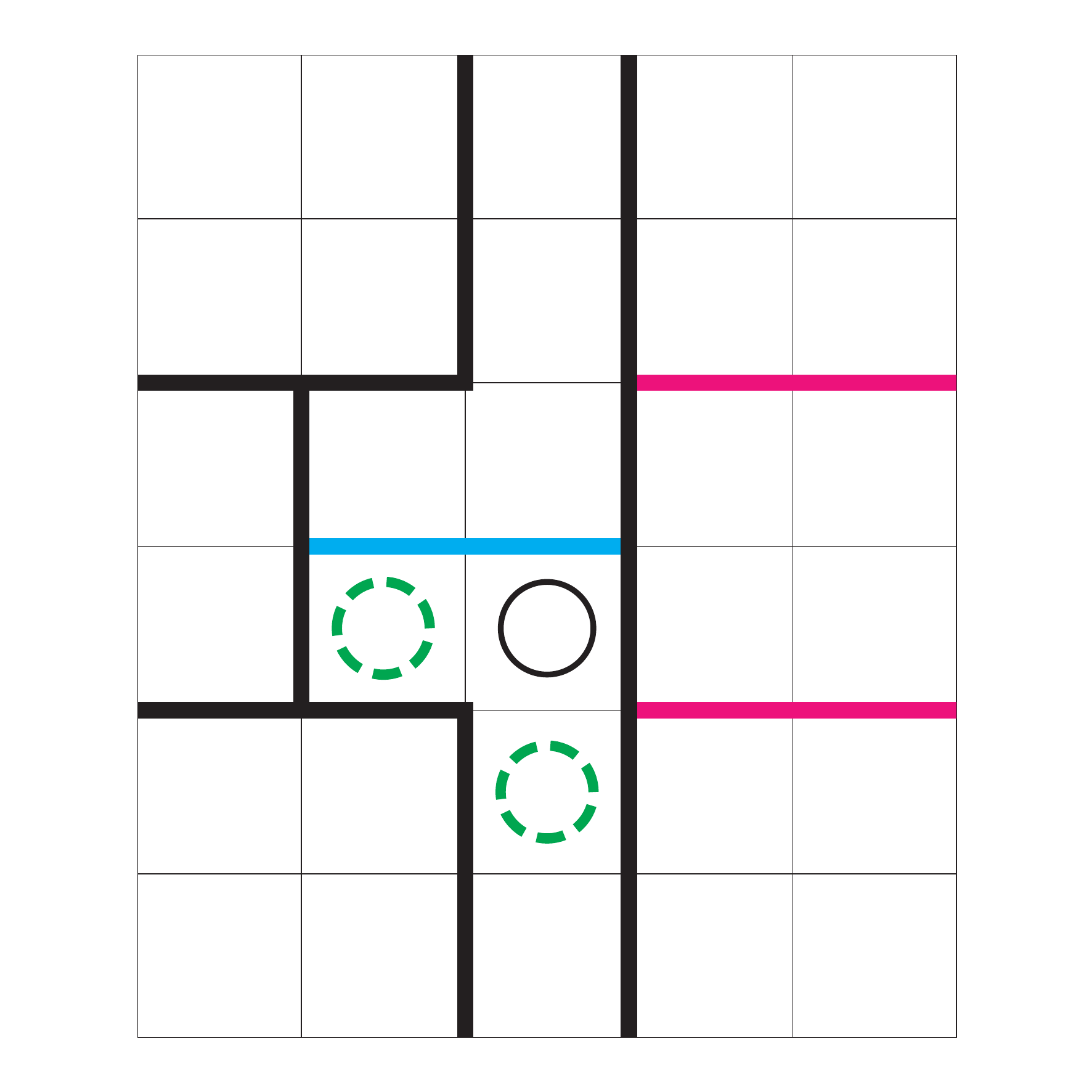}
    \caption{A horizontal wall placed 
    in the $2\times 2$ area. The gate is closed; upward movement is blocked.}
    \label{truth-value-assignment2}
    \end{center}
  \end{minipage}
  \hfill
  \begin{minipage}[b]{0.3\textwidth}
  \begin{center}
    \includegraphics[width=\textwidth]{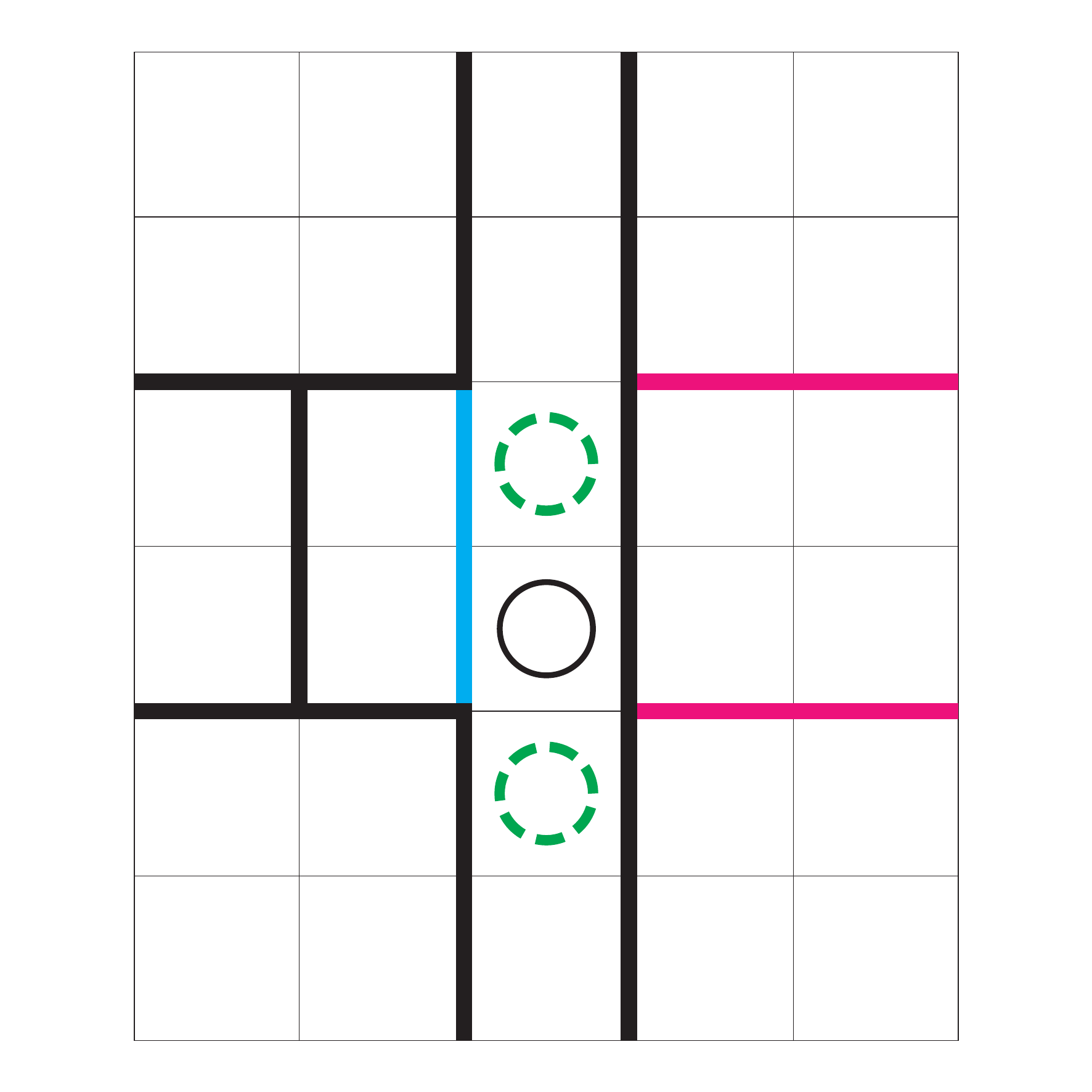}
    \caption{A vertical wall placed in the $2\times 2$ area. The gate is open; free vertical movement is guaranteed.}
    \label{truth-value-assignment3}
    \end{center}
  \end{minipage}
\end{figure}

Overall, the course of play thus follows four stages. First, White and Black alternate simulating the choosing of variables. Second, Black chooses a clause gadget for White's pawn to enter, by blocking off all other clause gadgets with walls except for one (which must remain accessible, under  
the rules). Third, White chooses a variable gadget for one of the variables within the chosen clause. The construction of the position allows White's pawn to force entry into the chosen variable gadget (in general, such entry is possible if and only if the variable is an element of the chosen clause). Fourth, White successfully escapes through the north exit of the chosen variable gadget if and only the exit is open (i.e., the variable is~{\tt true}). If so, White reaches her victory rank faster than Black can reach his. If not, White can only escape another way, which Black can delay indefinitely until such time as he has finished closing $b-1$ of his victory corridors, after which he exits the maze himself and wins the race.

\subsection{Proof details}

We now describe the position construction in
greater detail, 
then discuss the flow of play through the four stages given optimal play by both players.

\subsubsection{Position construction and terminology}
The \textbf{\textit{maze}} is the area within the large rectangle with corners labeled~7a and~7b in Fig.~\ref{fig:labeledreduction}, including the small eastern and southern \textbf{\textit{dead ends}}.  The \textbf{\textit{shortcut corridors}} are the vertical (horizontal) corridors within the maze that remain straight for the entire length (width) of the maze. Vertical (horizontal) shortcut corridors separate the maze columns (rows) that house the variable (clause) gadgets described below. Together, the shortcut corridors form a large rectangular grid whose cells each surround the intersection point of one clause-variable pair. 

The \textbf{\textit{clause gadgets}} are the horizontal corridors, straight save for the structures labeled~5, beginning from the right end of the chambers labeled 3 and terminating at the dead ends on the right of the maze. The \textbf{\textit{variable gadgets}} run vertically from the bottom of the boxes labeled 8 down to the dead ends at the bottom of the maze. Where no confusion arises, we may at times identify clause and variable gadgets with the corresponding clauses and variables in~$A$.

 An \textbf{\textit{elongator}} (Fig.~\ref{fig:elongator}) is a structure designed to lengthen White's travel through a corridor 
 by 
 directing the pawn in a zig-zag. A \textbf{\textit{winding road}}, labeled~5 in Fig.~\ref{fig:labeledreduction}, is a corridor for horizontal travel through a cell 
 containing one or more elongators to
 make
 walking through the winding road take $L$ more steps than if 
 it had been a straight corridor,
 where $L$  
 is a number
 greater than a maze  
 cell's perimeter.
 A \textbf{\textit{long and winding road}}, labeled~10, is similar but for vertical travel and with added length 3L. 

\begin{figure}
\centering
    \includegraphics[scale = 0.24]{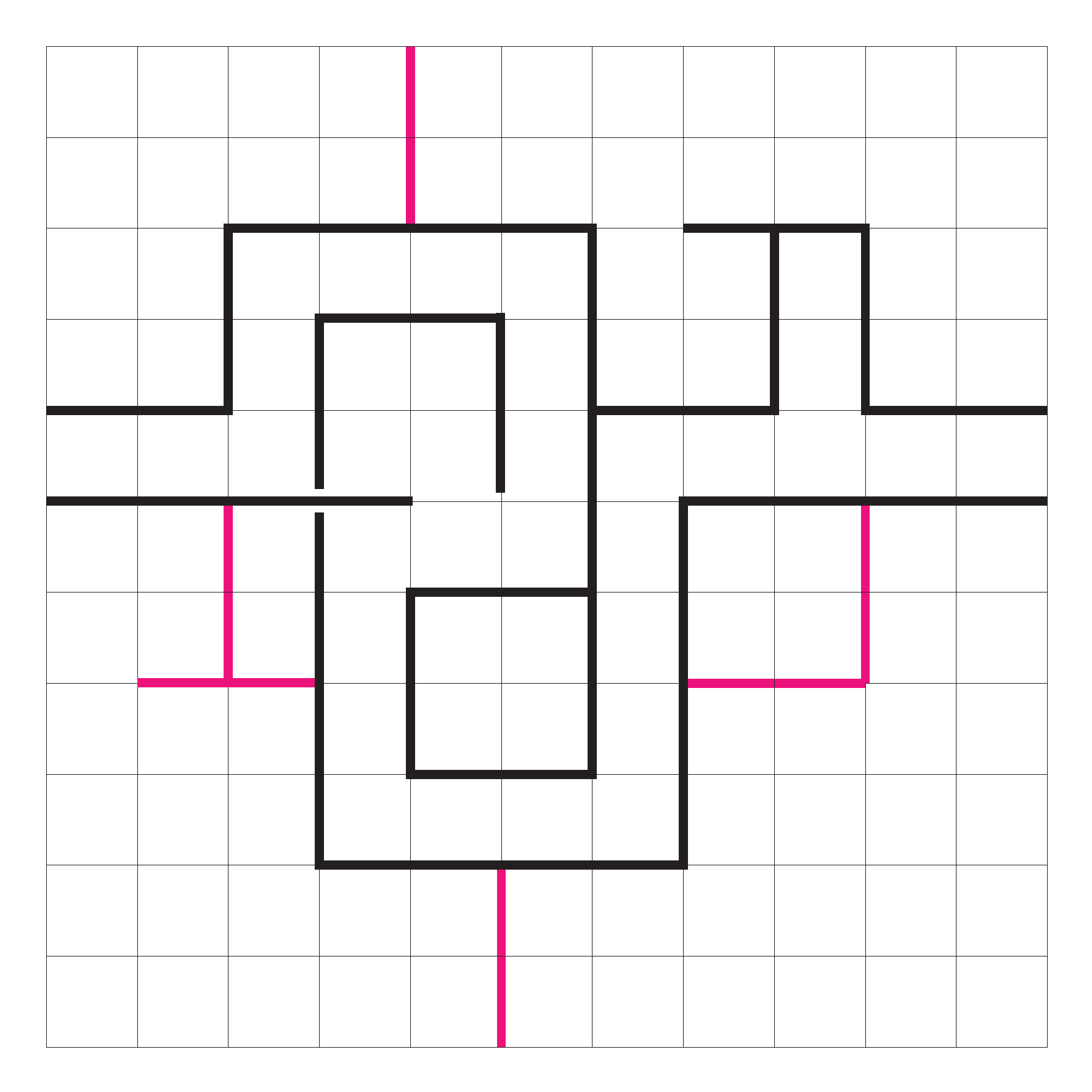}
    \caption{Example of an elongator. Size and shape can vary as needed.}
    \label{fig:elongator}
\end{figure}

A \textbf{\textit{side path}} (labeled~6), is a corridor giving passage from a clause gadget's corridor into a variable gadget's corridor.  
Inside the side path is a winding road. A side path from a clause gadget to a variable gadget is included when and only when the corresponding variable 
is in the corresponding clause 
in~$A$.  
An example of the absence of a side path---signifying that the clause does not contain the variable in question---can be found in the cell to the immediate right of the cell with the side path labeled~6 in the figure.  

We will see later that Black is unable to prevent White from entering a variable gadget from a clause gadget if and only if   
a side path is present between the two gadgets. Call the T-junction (oriented~$\dashv$) where a side path meets a variable gadget a \textbf{\textit{variable gadget entrance}}. Unlike  
all other T-junctions in the game position, variable gadget entrances are standard (blockable) T-junctions as in Fig.~\ref{fig:railroading-T2}, rather than unblockable T-junctions as in Fig.~\ref{fig:unblockable-T-junction}.

The \textbf{\textit{emergency exit}}, labeled 11, is an additional corridor out of the maze. This ensures that White always has 
an open path to victory once White is in the maze, as required under the rules. However, satisfying this formality is its only purpose, as  
Black can always 
keep  
White from using it by railroading unless Black is about to win the game.

The \textbf{\textit{black (white) victory corridors}}, labeled 12 (14), are vertical corridors to the bottom (top) rank. These are the only paths to victory for either player. White has one victory corridor with length~$24m$, 
whereas Black has~$b := 13m$  
victory corridors with length~$16m$
each. 
Each 
black victory corridor additionally contains a small (width~3) chamber at the bottom, labeled~13, where a wall can be placed to block it (see Fig.~\ref{fig:blackvictorycorridors}). So Black's corridors are each spaced three squares apart.

Note that the players' victory corridors together span the whole board's length, so the board height~$n$ is exactly $24m+16m=40m$. The board width is then 
also~$n = 40m$, which suffices
to contain the~$13m$ black victory corridors 
(whose total width is~$3\cdot 13m=39m$) 
with~$m$ files of room to spare for the maze's width~$m_w$ and the combined width of the external chambers and corridors in the west (all totaling below an additional~$m_w + m_w \ll m$).

\begin{figure}
\centering
\begin{subfigure}{\textwidth}
\centering
   \includegraphics[width=0.9\textwidth]{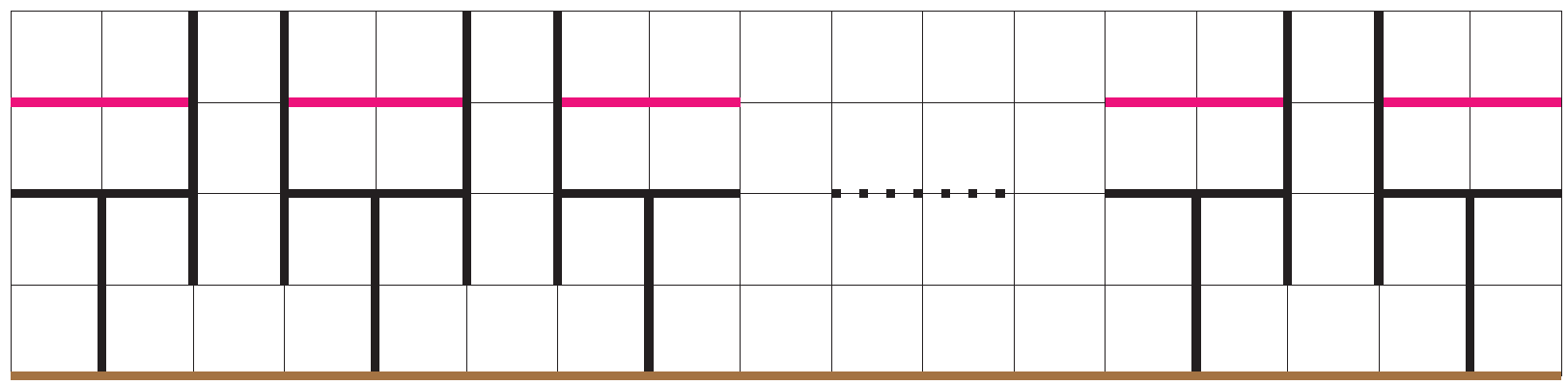}
   \caption{}
   \label{fig:blackempty} 
\end{subfigure}

\ 

\begin{subfigure}{\textwidth}
\centering
   \includegraphics[width=0.9\textwidth]{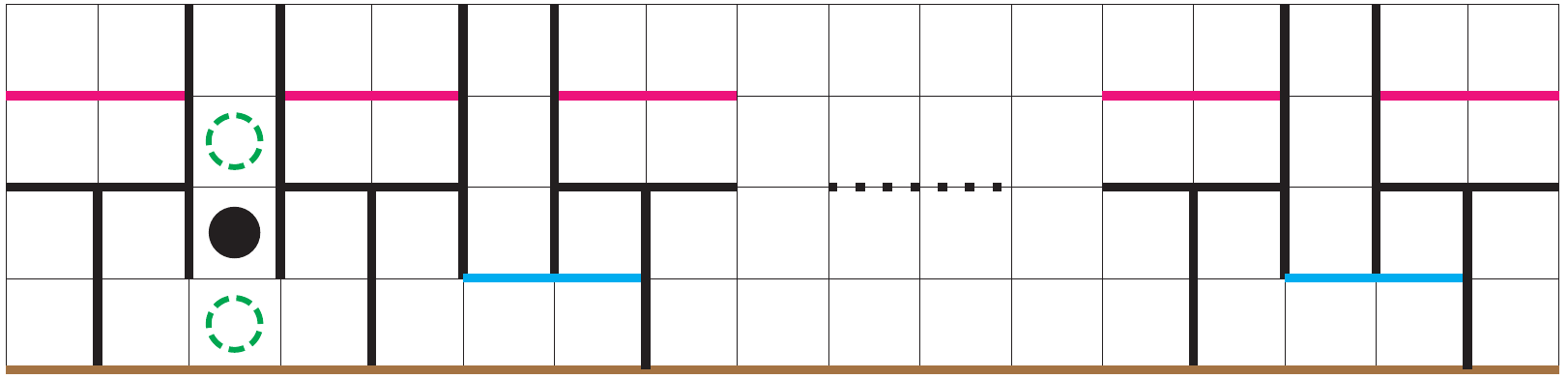}
   \caption{}
   \label{fig:blackfull}
\end{subfigure}
\caption[Label 13]{(a) The structure labeled~13 in the initial position. 
(b) The structure labeled 13 with walls placed so that the second and last black victory corridor are closed off. A black pawn is placed to show moving options.}
\label{fig:blackvictorycorridors}
\end{figure}

\FloatBarrier

The players' pawns stand at the locations indicated in Fig.~\ref{fig:labeledreduction}. Black has an abundance of wall pieces remaining in his supply---say,~$15m$---sufficient to be effectively unlimited. White has only~$\left\lceil\frac{v}{2}\right\rceil + 1$ walls---that is, one more than half the number of variables in~$A$, rounded up.\footnote{Our position, in which Black has so many more wall pieces remaining than White, is reachable from the starting position of an $(n \times n)$ Quoridor game if the vast majority of walls currently on the board have been placed there by White.} Of these walls,~$\left\lceil\frac{v}{2}\right\rceil$ are for simulating the selection of variables in $A$, with the one remaining piece saved for 
Black's victory corridors in the endgame.  
Because of this disparity in remaining wall pieces, Black is able to stay still in a junction to railroad White as long as desired ($b-1$ turns), while
White could never do the same in the reverse situation; if White tried, she would quickly run out of walls to place as Black simply jumps back and forth, leaving White no choice but to then make a pawn move.

To the west of each clause gadget entrance and in certain locations within each variable gadget are special \textbf{\textit{chambers}}, labeled~3 and~9 respectively, each of which can be blocked off by spending~$\left\lceil\frac{v}{2}\right\rceil + 2$ wall pieces. Only Black has sufficient walls to do this. The western chambers outside the maze are 
\textbf{\textit{clause entrance chambers}}, depicted in Fig.~\ref{fig:chamber}, west of which is a connecting vertical corridor labeled~4, called the \textit{\textbf{chamber access corridor}}. The horizontal distance from this corridor to the chambers is small, but larger than the number~$v$ of variables. Inside the maze, the chambers 
are identical to the clause entrance chambers, but rotated 
$90^\circ$ counterclockwise. 

\begin{figure}
\centering
    \includegraphics[scale = 0.33]{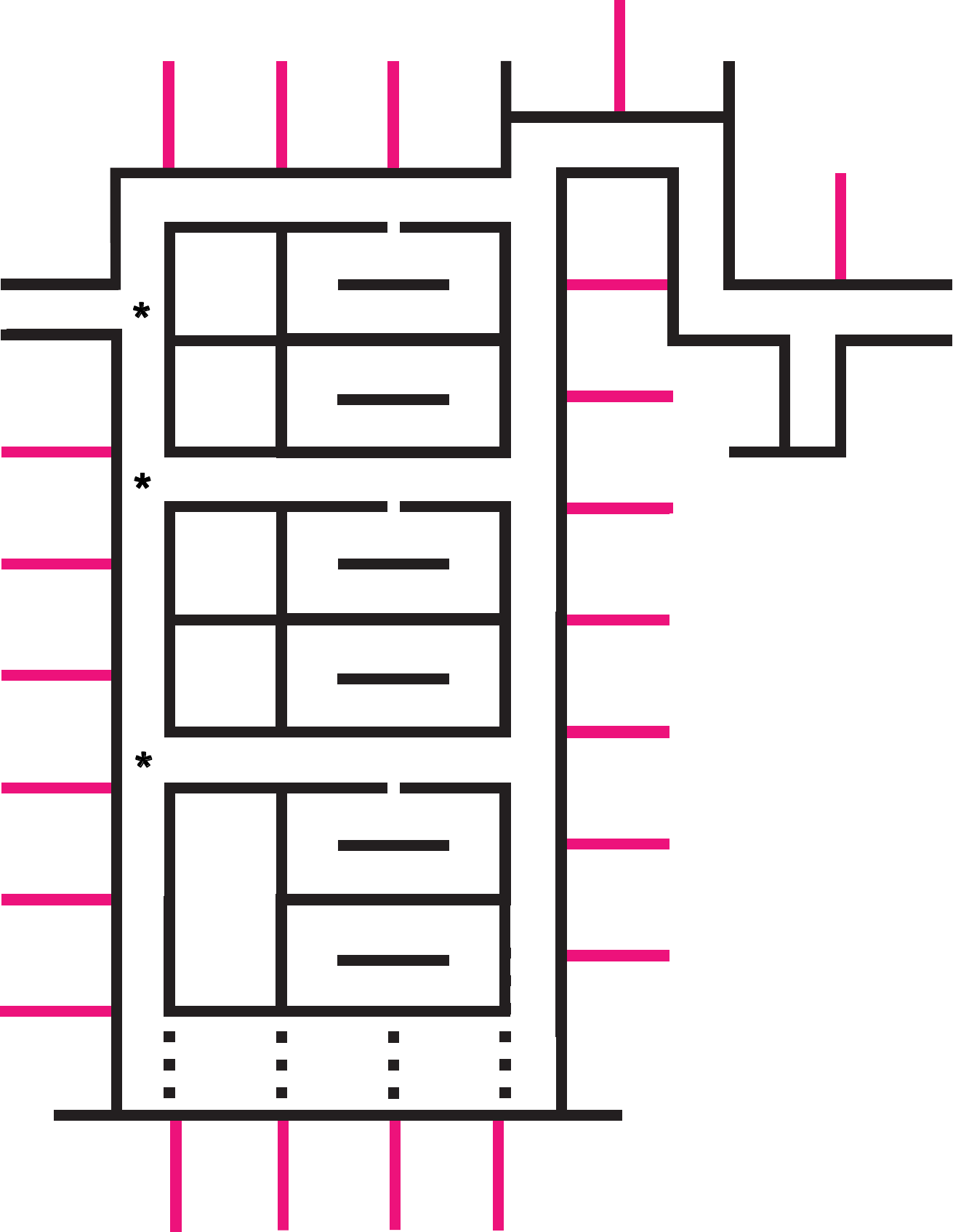}
    \caption{Simplified image of a clause entrance chamber. The actual gadget has unblockable T-junctions in place of the junctions marked (*), with 
    the rest of the gadget
    appropriately 
    scaled up in size 
    to make space. 
    Black can close the chamber by
    placing vertical walls in the open spaces (indicated by small white gaps) to block 
    off the horizontal corridors. The dotted lines denote a
    continuation of the pattern depicted 
    so that there are $\left\lceil\frac{v}{2}\right\rceil + 2$ horizontal corridors in total.
    }
    \label{fig:chamber}
\end{figure}

\begin{proposition}
     The position as a whole has size polynomial in the size of~$A$.
\end{proposition}

\begin{proof}
    Since the height (width) of a chamber outside (inside) the maze grows linearly with~$v$, the height~$m_h$ of the maze grows with~$vc$ and the width~$m_w$ grows with~$v^2$. So~$m=m_h m_w$ is determined by a polynomial function of~$v$ and~$c$. Since the board length~$n$ 
    is~$24m+16m = 40m$, the position size 
    is polynomial in the size of~$A$. 
\end{proof}

\subsubsection{Stages of play}
We now analyze
in detail
the four stages of play when both players play optimally.

In Stage~1, players simulate choosing variables, such that variables chosen by White become {\tt true} and ones chosen by Black become {\tt false}. This is done by placing walls in the~$2\times 2$ boxes above  variable gadgets in accordance with Figs.~\ref{truth-value-assignment1}-\ref{truth-value-assignment3}. 
Later, in Lemma~\ref{lem:stage1}, we will show that this play is optimal for both players.

    We remark that Stage~1 may end before all variables of $A$ are chosen, if the outcome of the \gpos\ simulation is already decided (i.e., every clause is already {\tt true} or one clause is already {\tt false}). Once it ends, Stage~2 is then characterized by White aiming to enter one of the clause gadgets in the maze via a clause entrance chamber. We now show that, in line with \gpos, it is Black (Player II) who chooses which clause gadget White's pawn enters.

\begin{lemma}\label{lem:White-maze-entry}
    White can always enter the maze from White's starting location. This takes fewer than $3m$ steps. However, White must enter the maze in a clause gadget of Black's choosing. 
\end{lemma}

\begin{proof}
    To enter the maze, White tries 
each of the~$c$ clause entrance chambers in order, starting from the top, until Black allows White to enter a clause gadget. (As later results will make clear, this is optimal, since if White can win, all clause gadgets offer a viable winning path through at least one true variable's exit, so White is happy to enter any gadget. And if White can't win, her decisions are irrelevant.) Within a chamber, White attempts passage east through each horizontal corridor
toward the clause gadget, successively from the top. Black can wall off each corridor in the chamber just before White passes through it (waiting until the last possible moment to do so), forcing White to backtrack  
through the corridor and enter the next corridor 
in the chamber. If and when the last corridor in a chamber is walled, White returns to the chamber entrance and proceeds down the chamber access corridor toward the chamber of the next clause. Recall that the T-junctions in the chamber access corridor are unblockable (as are all T-Junctions in the position, by default), so Black's pawn cannot railroad to prevent White from entering any of the chambers. Eventually, Black must let White exit east through one of the chambers, as walling off all paths to victory is illegal. 
Once this happens, Black can close the door behind White's pawn, preventing White from heading back to seek 
a different clause entrance.

Note that White's $\left\lceil\frac{v}{2}\right\rceil + 1$ walls are insufficient to close any chamber, so White has no influence on which chambers end up open or closed.

Also observe 
that White traverses no single square more than three times throughout  
this process.
Since the 
total 
size of the~$c$ chambers and the chamber access corridor is below~$m$ (in fact, the maze itself contains more than $v\cdot c$ further chambers inside it), entering the maze requires fewer than~$3m$ steps. 
\end{proof}

With White having entered a clause gadget~$c_i$, the game enters Stage~3. We now show that unless White uses a side path to enter a variable gadget, Black has a way of perpetually preventing White from ever moving vertically. 

\begin{lemma}\label{lem:stage3}
    After White enters a clause $c_i$, Black can reach every cross junction in~$c_i$ before White does and use railroading to prevent White from traveling vertically, unless White uses an available side path. Moreover, 
    Black can do this with time to spare, which he can
    spend to wall off excess victory corridors.\end{lemma}

\begin{proof}
    We use induction. For the base case, we show that Black can reach the first cross junction---the one at the entrance to the chose clause $c_i$---before White. Already from the start of Stage~2, after White's pawn first reaches the T-junction nearest to White's starting location, Black can copy any vertical move White makes. Whenever White walks horizontally into one of the corridors of a clause entrance chamber, Black can either use one move to close that corridor or let White walk through and close the corridor behind White. The moves Black spends placing walls in this way are fewer than the moves White spends moving horizontally. So Black reaches the first cross junction first, with time to spare.

    For the inductive step, we argue that if White moves from a cross junction~$j$ in the clause~$c_i$ to a neighboring cross junction in~$c_i$, Black can reach the neighboring cross junction first. When White begins moving horizontally away from~$j$, Black can first wait a turn and block off an excess victory corridor. After this, he races White by moving vertically downward, then horizontally through the nearest shortcut corridor, and back upward to reach the adjacent cross junction, all before White makes it across the one horizontal corridor in the clause gadget. This succeeds because White's passage through a winding road lengthens White's journey to~$L$ steps, which is greater than Black's route.\footnote{Note that the vertical leg of Black's journey that is not through a shortcut corridor passes straight through the chamber(s), demanding only very minimal 
    horizontal deviations---see Fig.~\ref{fig:chamber}, rotated~$90^\circ$. Unimpeded, Black passes through 
    the corridor in a nearly straight path. In principle, White could delay Black by placing walls in the chamber, but under normal circumstances, White only has one wall piece in stock after Stage~1. Even if White has more (because of the possibility, noted earlier, that Stage~1 ended before all variable exits were selected), White 
    never has the $\left\lceil\frac{v}{2}\right\rceil + 2$ needed
    to close the chamber completely, and the maximum delay White can impose on Black's vertical movement by placing walls here is, by construction of the position, still far  
    below~$L$ steps.} If White for some reason chooses midway to retrace steps and return to junction~$j$, Black can do the same, after which nothing has changed from the original situation except that Black is closer to walling~$b-1$ victory corridors.    
\end{proof}

While we now see that Black can railroad White away from turning north when there is no side path, the presence of side paths from~$c_i$ to all variables~$v_j \in c_i$ enables White to turn into any of the variable gadgets for such~$v_j$. To see that this is the case, recall that the T-junction at the entrance of a side path---the north end---is unblockable, unlike the T-junction oriented $\dashv$ at the exit on the east, which is also the variable gadget entrance. As seen in Figs.~\ref{fig:railroading-T1}-\ref{fig:railroading-T2}, the latter junction is blockable when, and only when, White wishes to travel into the stem, which is not the case at present.

We now proceed to Stage~4. In a similar vein to Lemma~\ref{lem:stage3}, we 
show that after White has chosen a variable gadget to enter, Black can prevent White from escaping and moving to 
a different gadget. Therefore, her 
only way out, if any, is through the exit above the~$2\times 2$ box.

\begin{lemma}\label{lem:stage4}
After White enters a variable gadget in Stage~3, Black can again reach junctions before White, with time to spare, and prevent her from leaving the gadget unless the north exit is open for White to escape the maze entirely.
\end{lemma}

\begin{proof}
    The only way out of a variable gadget, other than the northern exit, is to attempt to reach another gadget via a cross junction or variable gadget entrance (T-junction). Both of these can be neutralized by railroading, provided that Black can always reach them before White. We now show that this condition indeed holds.
    
    When White is in a variable gadget, the chambers are useless to White, since Black can close them off when White attempts to pass through them (as in the proof of Lemma~\ref{lem:White-maze-entry}). Therefore, the only route from one junction to another is through a long and winding road. This requires~$3L$ steps. To win the race from one cross junction to a vertically adjacent cross or T-junction, Black follows a similar method to that of Lemma~\ref{lem:stage3}. That is, to travel between two adjacent cross junctions, Black moves west, then vertically, then east, using only one winding road and two shortcut paths in total. This is below~$3L$ steps. To travel from a T-junction to a cross junction (the other direction is the same in reverse), Black 
    moves
    backward through the side path---west then north---then continues east or west along the clause gadget corridor, depending on which cross junction needs to be reached. East leads to the cross junction above the T-junction, while west leads to the shortcut paths 
    accessing 
    the 
    one
    below the T-junction. In both cases, Black 
    passes 
    through no more than two winding roads and a bit of additional distance, again totaling fewer than~$3L$ steps.
\end{proof}

The preceding results have laid the groundwork for proving Corollary~\ref{cor:sim-escape} and Lemma~\ref{lem:escape-Quor}, which will combine to show that Black wins the Quoridor game if and only if Player~II wins the simulated \gpos\ game, completing our proof by reduction.

\begin{corollary}\label{cor:sim-escape}
     Black can prevent White from escaping the maze indefinitely, if and only if Player~II is the winner of the \gpos\ simulation---i.e., formula~$A$ is false.
\end{corollary}
    
\begin{proof}
    Suppose~$A$ is false in the simulation. Then at least one clause~$c_i$ has all false variables---i.e., the north exits are closed in all variable gadgets of variables~$v_j$ to which a side path from~$c_i$ exists. By Lemma~\ref{lem:stage3} and the subsequent discussion, Black can prevent White from leaving~$c_i$ except through one of these side paths. When White uses such a side path to enter the gadget for~$v_j$, Black can prevent White's escape by Lemma~\ref{lem:stage4}, save for the north exit. But since the north exit is closed, White is trapped.

    Conversely, if the simulated~$A$ is true, then some gadget for a variable $v_j \in c_i$ is open, in which case Black cannot prevent White from entering the~$v_j$ gadget through its side path and walking north directly out of the maze.
\end{proof}

The following lemma is needed for the proof of Lemma~\ref{lem:escape-Quor}.

\begin{lemma}\label{lem:5m}
    In games 
    White 
    escapes the maze,
    her pawn
    travels from 
    its
    starting point to the victory corridor's entrance in under 
    $5m$ steps.
        In this case, since White's victory path has length $24m$, White threatens to win before turn $29m$.
    
\end{lemma}

\begin{proof}
By Lemma~\ref{lem:White-maze-entry}, White reaches a clause gadget in the maze for some clause~$c_i$ within~$3m$ steps. From there,  
optimal strategy 
is to walk directly to the side path leading to the entrance of the nearest variable gadget whose northern exit is open (i.e., the nearest true variable in~$c_i$). Then White enters the variable gadget and heads north to the exit, using  
the long and winding roads and not  
the chambers on the way (since Black could  
wall a chamber off during any attempt by White to pass  
through it, forcing White to backtrack). No square is traversed twice during this process, so this portion of the trip is under $m$ steps. Since the horizontal walk from the variable exit to the victory corridor is no greater than $m_w \ll m$ steps, the total journey is below~$3m+m+m_w \ll 5m$ steps.  
So, since White spends at most $\left\lceil\frac{v}{2}\right\rceil < v \ll m_w$ turns in Stage~1, White's victories all occur before turn~$\left\lceil\frac{v}{2}\right\rceil + 3m + m + m_w + 24m < 29m$.
\end{proof}

\begin{lemma}\label{lem:escape-Quor}
    Black wins the game if and only if he can indefinitely prevent White from escaping the maze.
\end{lemma}

\begin{proof}
    By Lemmas~\ref{lem:stage3} and~\ref{lem:stage4}, if Black can indefinitely prevent White from escaping the maze by railroading, then he can do so with time to spare that can be spent closing excess victory corridors. Eventually the penultimate corridor is closed, after which Black leaves the maze. This allows White to escape, but since Black's victory path is a full $24m-16m = 8m$ shorter than White's, Black wins.

    If Black cannot keep White inside the maze, White 
    reaches the top rank before turn~$29m$, by Lemma~\ref{lem:5m}. Black, however, needs in total~$13m-1$ turns to close victory corridors, several turns to approach the remaining victory corridor from his pawn's starting position (more than $m_w$), and 
    then~$16m$ 
    turns to walk down and reach the bottom rank, totaling over $(13m-1)+m_w+16m > 29m$. So White wins.
\end{proof}

Our desired Theorem~\ref{thm:PSPACE-hard} and Corollary~\ref{cor:main} now follow
directly from Corollary~\ref{cor:sim-escape} and Lemma~\ref{lem:escape-Quor}, once we provide the earlier promised proof that placing walls in the~$2\times 2$ boxes is optimal play during Stage~1. We do this in the following.

\begin{lemma}\label{lem:stage1}
    In 
    the first stage, White (Black) has no better moves than to simulate \gpos\ 
    by placing vertical (horizontal) walls in the $2\times 2$ boxes above the variable gadgets.
\end{lemma}

\begin{proof}
 To prove Lemma~\ref{lem:stage1}, we must 
 show that in Stage~1, no move is better for 
 White (Black) than simulating \gpos\ by placing a vertical (horizontal) wall in a $2\times 2$ truth-value assignment box. We consider White's perspective 
    first. 
    
    Recall that Stage~1 continues just until the outcome of the simulation is decided---i.e., when either a clause has all false variables or every clause has one true variable---which may occur before all variables are selected. During this time period, White has three alternatives to placing a vertical wall in a box: White can (a) place a horizontal wall in a box (which only helps Black), (b) place a wall in another location (which accomplishes nothing), or (c) move her pawn. We argue that while (c) technically accomplishes something, it is nothing relevant, as the time gained by moving the pawn closer to the maze is dwarfed in value by the missed opportunity to place a vertical wall in a~$2\times 2$ box. 
    Since the distance from White's starting location to the clause entrance chambers is greater than~$v$ (per the initial description of the chamber access corridor), White cannot ignore the simulation and reach the maze before Black has closed every variable exit, deciding the game. 
    And if White chooses to participate in the simulation on some turns while ignoring it on others, then for every turn she spends on a pawn move instead of placing a vertical wall, Black can place a horizontal wall instead, reducing this scenario to~(a). 

    From Black's perspective, moves other than placing a horizontal wall in a~$2\times 2$ box are pointless because the moves that make progress---walling off the victory corridors or moving the black pawn toward/through them---are insufficient to prevent White 
    winning if she can escape the maze, as the number of such moves needed exceeds~$29m$ (see Lemma~\ref{lem:escape-Quor}). No moves other than placing walls in~$2\times 2$ boxes impede White's escape at this stage; indeed, no other move accomplishes anything at all, except for closing one of the corridors in a chamber (which can always wait until White enters that corridor, if she ever does, in a future stage).
\end{proof}

This completes the main proof. $\hfill \qed$

\section{Discussion}\label{sec:disc}

We have shown that the decision problem for $(n\times n)$ Quoridor played with two players is PSPACE-complete. Moreover, Corollary~\ref{cor:main} shows this is the case even in positions in which it is guaranteed that one of the players has a forced win. This we proved via reduction from \gpos.

In Sec.~\ref{sec:results}, we claimed that our results obtain not only for Quoridor played with two players but also three- and four-player Quoridor as well. Quoridor with four players is played the same as the version with two, but with the additional two players facing east and west instead of north and south. For this reason, it is typically not recommended to play with exactly three players, as the lack of an opposite-facing opponent for one of the players creates an imbalance.

Justifying that  
$(n\times n)$ Quoridor played with more than two players 
is in PSPACE
is a straightforward adaptation of the proof in this paper. For 
hardness, our reduction can be modified to include the additional players but provide them with zero remaining wall pieces and with pawns located sufficiently remotely to have no chance to win or influence the game. For example, we may place a long horizontal chain of wall pieces just beneath the board's top rank, adding a small gap to create a T-junction with White's victory corridor, and place each new pawn in one of the two top corners. Since each needs~$n-1 = 40m-1$ moves to reach the opposite file and win---far longer than the first two players---their presence has no effect on the game result.

Future research may explore the hardness of games related to Quoridor, including Maze Attack, Pinko Pallino, and Blockade (or \textit{Cul-de-sac}). Maze Attack, the most recently published of these, is mostly identical to Quoridor, but with the crucial difference that pawns are not allowed to jump. That is to say, adjacent pawns prevent each other entirely from moving in that direction. Our present reduction is insufficient to prove hardness for this game, as Black simply wins automatically by outright blocking White from entering the maze until he is ready to win.

Pinko Pallino is a predecessor to Quoridor, released two years earlier by the same designer, Marchesi. A key difference is that pawns in Pinko Pallino can move diagonally. Because of this, the railroading strategy used at cross junctions in the present paper fails, so another reduction is needed.

Blockade, designed in 1975 by Philip Slater, is the oldest variant and was Marchesi's inspiration for Pinko Pallino and Quoridor. Players in this game have two pawns each. A turn consists of first moving a pawn two squares cardinally or one square diagonally, and then placing a wall, if any remain in supply. Jumping is permitted. Unlike other variants, Blockade has two sorts of walls, one of which may only be placed vertically and one only horizontally. Needless to say, this is sufficiently different from Quoridor that an alternative reduction is needed.

\section*{Acknowledgements}

We would like to thank Erik Stei, Kanae Yoshiwatari, Tonan Kamata, Robert Barish, and Daniele Muscillo for their helpful conversations and comments.

%
\bibliographystyle{plain}
 \bibliography{bibliography}
%

\end{document}